\documentclass[conference,onecolumn]{IEEEtran}
\IEEEoverridecommandlockouts
\usepackage{cite}
\usepackage{xcolor}
\usepackage{booktabs,url,enumitem,multirow,hyperref,tikz}
\usepackage{mathtools,amssymb,amsthm,algorithmicx,cleveref,csquotes,graphicx}
\usepackage{pgfplots,pgfplotstable,algorithm}
\usepackage[noend]{algpseudocode}
\pgfplotsset{compat=1.11}
\usetikzlibrary{automata,arrows,positioning,shapes}
\usetikzlibrary{decorations.text}
\usetikzlibrary{arrows.meta, shapes, positioning}

\newenvironment{myquote}[1]%
  {\list{}{\leftmargin=#1\rightmargin=#1}\item[]}%
  {\endlist}
  
\makeatletter
\def\BState{\State\hskip-\ALG@thistlm}
\makeatother
\let\oldReturn\Return
\renewcommand{\Return}{\State\oldReturn}
\algnewcommand{\Initialize}[1]{%
  \State \textbf{initialize:} \hspace*{0.2em}\parbox[t]{.4\linewidth}{\raggedright #1}
}

\renewcommand{\(}{\left(}
\renewcommand{\)}{\right)}

\newcommand{\vt}[1]{\boldsymbol{\mathbf{#1}}}

\newcommand{\qt}{\enquote}
\newtheorem{theorem}{Theorem}
\newtheorem{lemma}[theorem]{Lemma}
\theoremstyle{definition} 

\newtheorem{example}[theorem]{Example}
\newtheorem{remark}[theorem]{Remark}

\begin{document}

\title{Weighted Voting on the Blockchain: Improving Consensus in Proof of Stake Protocols\thanks{This work was supported in part by the National Research Foundation (NRF), Prime Minister's Office, Singapore, under its National Cybersecurity R\&D Programme (Award No. NRF2016NCR-NCR002-028) and administered by the National Cybersecurity R\&D Directorate. Georgios Piliouras acknowledges SUTD grant SRG ESD 2015 097, MOE AcRF Tier 2 Grant 2016-T2-1-170 and NRF 2018 Fellowship NRF-NRFF2018-07.}
}

\author{
\IEEEauthorblockN{Stefanos Leonardos\IEEEauthorrefmark{1}, Dani\"el Reijsbergen\IEEEauthorrefmark{1}, Georgios Piliouras\IEEEauthorrefmark{1}}
\IEEEauthorblockA{\IEEEauthorrefmark{1}Singapore University of Technology and Design}
}

\maketitle
\begin{abstract}
Proof of Stake (PoS) protocols rely on voting mechanisms to reach consensus on the current state. If an enhanced majority of \emph{staking nodes}, also called validators, agree on a proposed block, then this block is appended to the blockchain. Yet, these protocols remain vulnerable to faults caused by validators who abstain either accidentally or maliciously. \par To protect against such faults while retaining the PoS selection and reward allocation schemes, we study \emph{weighted voting} in validator committees. We for\-malize the block creation process and introduce validators' \emph{voting} \emph{profiles} which we update by a \emph{multiplicative weights} algorithm relative to validators' voting behavior and aggregate blockchain rewards. Using this framework, we leverage \emph{weighted majority voting rules} that optimize collective decision making to show, both numerically and analytically, that the consensus mechanism is more robust if validators' votes are appropriately scaled. We raise potential issues and limitations of weighted voting in trustless, decentralized networks and relate our results to the design of current PoS protocols.
\end{abstract}


\section{Introduction}\label{sec:introduction}

In the Nakamoto consensus protocol \cite{Na08} that underpins the popular Bitcoin cryptocurrency, a single miner claims the right to append the next block to the blockchain by a demonstrating so-called Proof of Work (PoW), i.e., the solution to a moderately hard cryptographic puzzle. The strengths and vulnerabilities of this mechanism are well understood, see \cite{Ga15,Leo16,Ga17} and \cite{Ey14,Ki16,Sa17}. Two fundamental properties satisfied by PoW selection are the following: (i) Miners holding $p\%$ of computational power will create on average $p\%$ of blocks that will be part of the blockchain assuming that all miners follow the protocol. This property relies on the randomness of the cryptographic puzzle and ensures \emph{fairness} in the allocation of mining rewards\cite{Pa17a,Pa17b}. (ii) Miners can be selected as the next block creators only when they actually create and submit the block to the network. This property asserts that the random selection of the next block creator and the creation of the next block occur simultaneously. \par
Proof of Stake (PoS) or virtual mining protocols \cite{Bon15} constitute alternative selection mechanisms that aim to retain PoW's benefits while improving on its weaknesses \cite{Be16,Br18}. While different PoS protocols propose different schemes\cite{ouro17,Gi17, Be16b,Pa18, Bu1810,Kw14}, in general, the block proposal mechanism is the following: blocks are created by staking nodes, also called \emph{validators}\cite{Kw18,Da18}, who are granted the right to participate in the block creation process by locking capital in protocol currency, called \emph{stake}. Subsequently, a pseudo-random mechanism selects nodes proportionally to their stake to form committees that will decide on the validity of a block proposed by a selected node, called \emph{block proposer} or \emph{leader}. At the core of the consensus mechanism, the selected validators cast approval or disapproval votes on the proposed block. If an enhanced majority of approval votes, usually $2/3$ of the committee size \cite{La82}, is reached, then the proposed block is accepted and appended to the blockchain. \par
However, maliciously or accidentally abstaining validators can cause consensus failures and stall the blockchain\footnote{Malicious absention refers to non-voting or incorrect voting to attack the protocol, whereas accidental abstention refers to a variety of reasons, such as dropping offline, experiencing network latency, bugs in client updates or being a victim of a censoring/eclipse attack.}. 
The reason is, that unlike Nakamoto consensus \cite{We18}, PoS protocols decouple the selection of block creator(s) from the actual block creation and hence do not satisfy property (ii) above. Since consensus on the \qt{valid} state of the blockchain relies on the voting behavior of all participating validators, this also implies that the actual rate of blocks that a validator gets to create may differ from the times that they get selected in a committee. Hence, while the PoS selection mechanism aims to ensure fair allocation of \qt{mining} rewards in line with PoW's property (i) above, in practice, it may fail to do so. These observations motivate the following question:
\begin{myquote}{0.15in}
\emph{How can we improve the efficiency and robustness of the consensus mechanism, i.e., how can we use information on validators' past voting patterns to enforce that with overwhelming probability the selected validator committees will reach consensus on the \qt{valid} state of the blockchain?} 
\end{myquote}
The problem of improving consensus has been treated in the context of fixed-size committee voting by a rich stream of literature\cite{Ni82,Ni84,Sh84,Yo95,Be97}. The derived solutions depend on quantifying voters' abilities to make correct decisions and apply \emph{weighted voting rules} in which votes are weighted according to voters' profiles. Our goal is to apply these results in the setting of PoS protocols. The link is immediate, since, as remarked in \cite{Be97}, the assumptions of their original model imply both decentralized information processing and limited communication. The additional challenges that have to be treated in the blockchain context, concern the updating of these profiles and protection against their manipulation by adversaries. \par
To address this problem, we formulate the proper mathematical framework and develop a model to quantify validators' \emph{voting profiles}. The proposed scheme is applied \emph{once} voting committees have formed and does not modify the underlying PoS selection and reward mechanisms. Each staking node (validator) is assigned a score based on their so-far contribution to protocol execution. When selected for a committee, their vote is weighted relative to their profile and consensus is decided according to a \emph{weighted majority rule} that maximizes the expected collective rewards \cite{Sh84,Be97}. Finally, based on their vote and the overall consensus outcome, their voting profile is revised according to a fully parametrizable \emph{multiplicative weights update algorithm} \cite{Ar12,Ba18}. Supported by numerical examples and simulations, our findings demonstrate that weighted voting renders the consensus mechanism more efficient, even if more than $1/3$ of nodes are not properly voting. In this way, the proposed scheme restores fairness without compromising other PoS features. Additionally, since it does not modify the underlying PoS mechanism, it can be tested, implemented and reverted with minimal cost to existing protocol users. \par
We raise issues that pertain to weighted voting such as loss of anonymity \& centralization and discuss their relevance to protocol design and implementation. 
\subsection{Related Work}\label{sub:related}
Although weighted voting in distributed systems is known to increase efficiency, incentive compatibility \cite{Az07,Az14}, and network reliability \cite{Ba08b}, it also raises additional risks \cite{Ba08a,Az09,Zu12}. Similar to \emph{proof of reputation} systems \cite{Yu18}, weighted voting deviates from the principle \emph{one node -- one vote} and hence, is vulnerable to manipulation by adversaries. Weighted voting becomes particularly relevant in less anonymous \cite{Ko14}, private or permissioned blockchains \cite{An18,St18,Vi18}, in \emph{delegated} PoS protocols \cite{Gr17,Li18} and in PoS protocols with low targeted number of \emph{staking pools}\cite{Bru18}. However, under rather general assumptions it can also be relevant for permissionless Proof of Stake or hybrid Proof of Work and Proof of Stake systems \cite{Bu1810,ouro17}. We provide such a use case in \Cref{sub:use1}.\par
A profiling system reduces or eliminates the anonymity of staking nodes which is at odds with the design philosophy of permissionless blockchains \cite{Pre20,Oce20}. However, with blockchain governance yet to be determined, introducing less anonymity may be a desired feature. Recent results support relatively low desired numbers of staking nodes \cite{Bu1811} or stake pools \cite{Bru18}. In such schemes, reputation will implicitly or explicitly influence protocol execution. Moreover, stake pools can retain anonymity at the user level, i.e., while the pool itself becomes identifiable by its voting profile, the users remain anonymous. In any case, the introduction of voting profiles and weighted voting seems particularly relevant to permissioned blockchains or delegated PoS mechanisms.\par
More interesting is the implementation of the weighted voting algorithm in conjunction with state of the art consensus schemes such as the delegated Proof of Stake of EOS.IO \cite{Gr17}. Instead of substituting the underlying consensus mechanism, the weighted voting scheme enhances its functionality and operability by providing an additional layer of defense against accidental or adversarial system failures. In fact, in consensus systems that use a limited number of delegates, weighted voting can prove beneficial in accelerating consensus and enhancing their scalability, liveness and safety. This rests on the fact that assigning and updating profiles to the delegates is much easier than in general permissionless blockchains. However, even in the latter case, the expected latency (and overhead) that will be introduced to the system by the profiling system does not exceed nor significantly increases the computational complexity of updating the validators' stakes.\par
Finally, claiming that \emph{all} PoS protocols fit under the stylized model -- see \Cref{sec:model} -- that we use to design the proposed weighted voting scheme would be an oversimplification and wrong. Yet, most PoS protocols that we are aware of involve voting mechanisms and hence, may benefit -- to a larger or lesser extent -- from the present proposal. An incomplete list includes \cite{ouro17,Be16b,Gi17,Pa18,Kw14}. For more extensive details of PoS protocols, we refer to \cite{Br18, Ba17,Di18,Ca17}. Further comparisons to existing proposals and implementation details are given in \Cref{sec:implementation}

\subsection{Outline}
\label{sub:outline}
In \Cref{sec:model}, we abstract the PoS consensus mechanism as a voting game. In \Cref{sec:results}, we construct the improved voting scheme and illustrate it with examples. \Cref{sec:mwua} is devoted to the design of the updating algorithm and the derivation of bounds on the validators' faulty behavior that can be tolerated with significantly affecting their profiles. The resulting scheme -- weighted voting rule and multiplicative weight updating algorithm -- are tested and related in potential use cases in \Cref{sub:numerical}. We conclude the paper with a discussion about limitations and implementation issues in \Cref{sec:implementation} and a summary of our findings in \Cref{sec:conclusions}.\par
As an extension of a conference paper\cite{Le19}, the present article contributes both technical and application specific materials that make the paper self-contained and enhance the intuition behind the derived results. In addition to a detailed comparison with related work, cf. \Cref{sub:related}, technical extensions comprise \Cref{lem:work} and its proof, the proof of \Cref{thm:main}, \Cref{fig:work} and \Cref{rem:work}. From an application oriented perspective, the present paper elaborates on the design and optimal parametrization of such voting schemes in practice. This is done in \Cref{sec:faulty}, \Cref{sub:use1} and \Cref{sub:use2} in the context of the state of the art PoS proposals and implementations, such as Ethereum with Casper FFG and EOS.IO. The main focus of the analysis is in the mechanics of committee formation and the risk mitigation that is achieved via the proposed scheme for potential investors. In particular, the new content addresses incentives and adversarial behavior in potential use cases and moves forward to explain the steps for the adoption of the proposed scheme in practice. 

\section{The Model}\label{sec:model}

Our terminology is based on the Ethereum 2.0 PoS protocol specification\cite{Bu1810}. Yet, our model remains as general as possible in an effort to capture similar voting mechanisms implemented by related PoS platforms.
\begin{description}[leftmargin=0cm]
\item[Time:] Time is divided into \emph{time slots} $t\in\{0,1,2,\dots\}$ of fixed duration $d$. Each time slot is dedicated to the proposal and creation of a new \emph{block} $B_t$. Time slot $t=0$ is the time of creation of the \emph{genesis block} $B_0$. 
\item[Validators:] The main actors in the block proposal and creation mechanism are the \emph{staking nodes}, also called \emph{validators}, denoted by $i\in I$, where $I\subseteq \mathbb N$ is the \emph{set} of all validators. $\mathcal B_{i,t}$ will denote the set of blocks for which validator $i$ is aware of at time $t\ge0$. 
\item[Stake:] The deposit or \emph{stake} is the amount of the underlying cryptocurrency that a potential validator locks as \emph{Proof of Stake} (PoS) to participate in the block creation process. Such deposits may change over time. Accordingly, let $v_{i,t}\ge0$ denote the stake of validator $i\in I$ and $v_t:=\sum_{i\in I}v_{i,t}$ the total stake at time slot $t\in\mathbb N$. If $v_{i,t}>0$, then validator $i$ is called \emph{active} at time slot $t$. Validators who have withdrawn from $I$ or who have not entered it yet, can be thought of as validators with stake $v_{i,t}=0$. Thus, although the set of validators is dynamic, we may write $I$ instead of $I_t$ to denote the set of validators at any time $t\ge0$.
\item[Block proposer \& committees:] To create blocks and extend the blockchain, active validators are selected \emph{proportionally to their stake} by a pseudo-random mechanism which assigns to each time slot $t$ a \emph{leader} or \emph{block proposer} and a fixed-sized \emph{committee} $N_t=\{1,2,\dots,n\}$ of validators\footnote{The mechanics of the pseudo-random mechanism vary between different protocols. Here, we are not interested in risks associated with manipulating this mechanism and focus on the mechanics of the voting process \emph{once} a random committee has been formed.}. The block proposer is assigned the task to propose a block $B_t$ to the committee. In turn, the committee votes on whether the proposed block should be appended to the blockchain or not. This process constitutes the core of the consensus mechanism and will be the focus of the present paper. 
\item[Validators' strategies:] The set of strategies of a validator who has been selected in a committee will be denoted by $S=\{-1,1\}$, where $-1$ stands for rejecting and $1$ for approving the proposed block. In particular, $-1$ also corresponds to \emph{not} casting a vote, either deliberately or accidentally. Accordingly, let $X_{i,t}$ denote the indicator random variable

\[X_{i,t}=\begin{cases}\phantom{-}1, &\text{if validator $i$ voted on the approval of the proposed block $B_t$ in time slot $t$}\\-1, & \text{otherwise}\end{cases}\]

We will write $\mathbf X_t:=\(X_{1,t},X_{2,t},\dots,X_{n,t}\)$ to denote the random vector of all indicator random variables \emph{prior to} the validators' voting decisions and $\mathbf x_t:=\(x_{1,t},x_{2,t},\dots,x_{n,t}\) \in \{-1,1\}^n$ to denote the realized \emph{decision} or \emph{action} profile of the committee at time $t\ge0$. Accordingly, we will write $X:=\{-1,1\}^n$ to denote the set of all possible decision profiles.
\item[Decision rule:] A \emph{decision rule} $f:X\to \{-1,1\}$, also called \emph{social choice function} or \emph{aggregation of preferences rule}, is a function that receives as input the action profile $\mathbf x_t$ and outputs a decision in $\{-1,1\}$, where $-1$ and $1$ stand for disapproval and approval of the proposed block, respectively. We will focus on \emph{(simple or enhanced) majority rules} $f_q, q\in[0.5,1]$ defined by 

\begin{equation}\label{eq:decision}
f_q\(\mathbf x_t\):=\begin{cases}\phantom{-}1, &\text{if }\sum_{i=1}^nx_{i,t}\ge \(2q-1\)n,\\ -1, & \text{otherwise}\end{cases}
\end{equation}

If at least a fraction $q\in [0.5,1]$ of the selected validators approve the proposed block, then this block is appended to the blockchain. Otherwise, the time slot remains empty and the mechanism progresses to the next time slot.
\end{description}

If block proposers follow the protocol and do not behave maliciously, then all proposed blocks are valid. Moreover, if the network is not partitioned and network latency is insignificant (lower than the time slot duration during which votes are expected to appear), then there is no controversy about which blocks are valid, since all validators view essentially the same blockchain (state), i.e., $\mathcal B_{i,t}:=\mathcal B$ for all $i\in I, t\ge0$. Under these conditions, the required majority should be reached and valid blocks should be regularly approved and appended \cite{Gi17,Pa18}. However, in practice, two main reasons may lead to failures on the consensus mechanism 
\begin{itemize}[leftmargin=*]
\item adversarial behavior: a malicious node abstains from voting or censors other validators' votes to block the required majority and stall the block creation process. 
\item accidental behavior: validators drop offline accidentally or due to negligence, they experience bugs on client updates or bad network connectivity, their votes do not propagate through the network in the expected time slot or they are victims themselves of a censoring/eclipse attack.
\end{itemize}
Our goal is to study how existing results on optimizing aggregation of preferences in committees, in particular the results of \cite{Ni82,Sh84,Be97}, can be applied to the PoS blockchain setting and improve the underlying consensus mechanism. The application of these results will be immediate once we have defined the proper framework. 

\section{An Improved Voting Rule}\label{sec:results}

To optimize the consensus process from an aggregative perspective, we quantify the collective benefits and losses (payoffs) from correct and wrong decisions respectively. This is done in \Cref{tab:welfare}. 
\begin{table}[!htb]
\centering
\begin{tabular}{@{}lrrc@{}}
& & \multicolumn{2}{r}{\textbf{Proposed Block} $B_t$ }\\[0.2cm]
&& Valid (1) & Invalid (-1)\\
\cmidrule{2-4}
\multirow{2}{*}{\textbf{Committee\hspace{4pt}}}& Approve & 1 & $-\ell_a$\\
\cmidrule{2-4}
& Reject & $-\ell_r$ & $\hspace{1em}1$\\
\cmidrule{2-4}\\
\end{tabular}
\caption{Collective welfare from consensus outcome.}
\label{tab:welfare}\vspace*{-0.3cm}
\end{table}
The benefit from making a correct decision, i.e., approving a valid block or rejecting an invalid block, is scaled to $1$. Here, $-1$ denotes an invalid and $1$ a valid block $B_t$, cf. \Cref{sub:valid}. If a valid block is rejected, then a loss of $\ell_r>0$ is incurred which corresponds to the waste of computational resources and the failure of the system to process pending transactions. On the other hand, if validators vote for an invalid block, then $\ell_a>0$ represents the losses from validating a conflicting history\footnote{This may include approval votes for a block that a malicious node is trying to create in order to double-spend or perform some other kind of attack. It may also involve votes that get wasted on blocks that will be subsequently reverted or abandonded.}. Determining the exact values of $\ell_r$ and $\ell_a$ is a matter of protocol parametrization. In the present treatment, and without compromising the generality of the results, we will assume that $\ell_r\ll \ell_a$, i.e., that accepting invalid blocks has greater repercussions for the system, since such blocks typically stem from malicious -- and not accidental -- behavior. Finally, let $\alpha\in\(0,1\)$ denote the prior probability that a proposed block is \emph{invalid}, e.g., blocks that an adversarial is trying to create.\par
Given the above, we seek to maximize the \emph{expected collective welfare} $E_t$ at time slot $t>0$. $E_t$ depends on the probabilities of accepting a valid block and of rejecting an invalid block under the decision rule $f_q$, 

\begin{equation*}\pi_{1}\(f_q\):=P\(f_q\(\mathbf X_t\)=1\mid B_t=1\)\end{equation*}
and $\pi_{-1}\(f_q\):=P\(f_q\(\mathbf X_t\)=-1\mid B_t=-1\)$, respectively. Using this notation, 

\begin{equation}\label{eq:welfare}E_t\(f_q\)=\(1-\alpha\)\(1+\ell_r\)\pi_{1}\(f_q\)+\alpha\(1+\ell_a\)\pi_{-1}\(f_q\)\end{equation}
where we have omitted constant terms. To estimate $\pi_{1}\(f_q\)$ and $\pi_{-1}\(f_q\)$, and hence, to maximize $E_t$, we need to reason about the decision rule $f_q$ and particularly, about the distribution of the decision variables $X_{i,t}$. Fortunately, this can be done by retrieving existing information about validators' past votes that have been stored as messages on the blockchain. This is captured by the notion of validators' \emph{voting profiles} that we introduce next.

\begin{description}[leftmargin=0cm]
\item[Voting profiles:] Each validator $i\in I$ is assigned a \emph{score} $p_{i,t}\in[0,1]$ that corresponds to their \emph{voting profile} at the start of time slot $t$. The value $p_{i,t}$ can be thought of as validator $i$'s \emph{decision ability} or \emph{probability} that $i$ will vote correctly, i.e.,  

\begin{equation}\label{eq:correct}p_{i,t}:=P\(X_{i,t}=1\mid B_t=1\)=P\(X_{i,t}=-1\mid B_t=-1\)\end{equation}

for $i\in I, t\ge 0$. For instance, in its simplest form, $p_{i,t}$ can be thought of as the fraction of validators $i$'s correct votes to the number of slots in $\{0,1,\dots,t-1\}$ that $i$ was selected in a committee. In what follows, we will examine a more elaborate scheme to define the profiles $p_{i,t}$ that depends on the collective welfare of the consensus outcome, cf. \Cref{tab:welfare}.
\item[Initializing \& suspending profiles:] We will set a newly entering validator $i$'s voting profile at $p_{i,0}:=0.5$ and will require that $p_{i,t}\in [0.5,1)$, for any $t\ge g$, where $g\ge1$ denotes an initial grace period. If $p_{i,t}<0.5$, for some $t\ge g$, then validator $i$ will be suspended from $I$. The reasoning behind these choices is detailed in \Cref{sub:initiate}. 
\item[Updating scheme:] In general, an \emph{updating scheme} is given by a function $h: [0,1] \times \{-1,1\}\to [0,1]$ which revises validator $i$'s voting profile after time slot $t$ based on $i$'s prior voting profile $p_{i,t}$, the correctness of their voting decision $x_{i,t}$ and the current state $\mathcal B_t$, i.e.,
\[p_{i,t+1}=h\(p_{i,t},x_{i,t}, \mathcal B_t\)\]
for $t>0$. The current state $\mathcal B_t$ may include all relevant information, such as the collective welfare parameters, cf. \Cref{tab:welfare}, the validity of the proposed block and the consensus outcome. If a validator $i\in I$ has not be selected for slot $t$, then simply $p_{i,t+1}=p_{i,t}$. A concrete updating scheme that fits in this description is developed in \Cref{sec:results}.
\end{description}

Given the introduced validators' voting profiles, we now return to the problem of maximizing the collective welfare $E_t$.

\begin{lemma}\label{lem:work}
For a selected validator committee $N=\{1,2,\dots,n\}$ at time slot $t$, we may condition on the vector of validators' voting profiles $\mathbf p_t=\(p_{1,t},p_{2,t},\dots,p_{n,t}\)$ and write the probability $\pi_{1}\(f_q\mid \vt p_t\)$ of approving a correct block with the decision rule $f_q$ as
\begin{equation}\label{eq:probability}\pi_{1}\(f_q\mid \vt p_t\)=\sum_{\vt x_t:f_q\(\vt x_t\)=1}\(\prod_{i:\\ x_{i,t}=1}p_{i,t}\prod_{j: x_{j,t}=-1}\(1-p_{j,t}\)\)\end{equation}
and similarly for $\pi_{-1}\(f_q\mid \vt p_t\)$. 
\end{lemma}
\begin{proof}
This expression for $\pi_{1}\(f_q\mid \vt p_t\)$ is derived as follows: The summation ranges over all action profiles $\vt x_t$ for which the decision rule $f_q$ approves the proposed block, i.e., $f_q\(\vt x_t\)=1$. The double product inside the parenthesis is precisely the likelihood of each of these profiles given that validators' decisions are independent. Specifically, the first product ranges over all validators $i\in I$ who vote correctly, i.e., $i:x_{i,t}=1$, and multiplies each one's probability of a correct vote, i.e., $p_{i,t}$, and the second ranges over the remaining validators $j\in I$ who vote incorrectly, i.e., $j:x_{j,t}=-1$, and multiplies each one's probability of voting incorrectly, i.e., $(1-p_{j,t})$. 
\end{proof}

Given these expressions, the problem of maximizing $E_t\(f_q\)$ can now be studied as an instant of the \emph{committee-voting} models in \cite{Be97} and \cite{Ni82,Sh84}. 

\begin{example}[\hspace{1sp}Adjusted from \cite{Sh84}]\label{ex:shapley}
Consider a committee of $5$ validators with voting profiles, i.e., empirical probabilities of voting correctly, $\vt p=\(0.9,0.9,0.6,0.6,0.6\)$, as in \cite{Sh84}. In the unweighted case, or equivalently in the case in which all votes are weighted equally, and under the $2/3$-majority decision rule $f_{2/3}$ that is commonly used in consensus protocols \cite{Pe80,La82}, the probability of reaching consensus on the correct block is equal to the probability that at least $4$ out of the $5$ validators vote correctly. This is 

\begin{align*}0.9^20.6^3+2\cdot0.6^3\(0.9\)\(0.1\)+3\cdot0.9^20.6^2\(0.4\)\approx 0.56
\end{align*}
which is lower even than the lowest voting profile of $0.6$. A naive improvement would be to only consider the vote of the validator with the highest voting profile, i.e., of either the first or the second validator. This would increase the probability of correct voting to $0.9$, however at a toll on decentralization. 
\end{example}

In the naive improvement of the previous example, the vote of the best validator received a \emph{weight} of $1$ and the vote of all others a \emph{weight} of $0$. This raises the question of whether we can assign non-trivial \emph{weights} (scale factors) to all validators and still improve the probability of a correct decision. The answer is affirmative and hinges on the notions of \emph{weighted voting} and \emph{weighted majority rules}.

\begin{description}[leftmargin=0cm]
\item[Weighted majority rule:] For a set of $n$ validators, let $\vt w_t:=\(w_{1,t},w_{2,t},\dots,w_{n,t}\)$ denote a vector of non-negative \emph{weights} or \emph{scaling factors}, with $w_t:=\sum_{i=1}^nw_{i,t}$. The \emph{weighted majority rule}, $f_{q}\(w\)$, or simply $f_q$, is defined as 

\begin{equation}\label{eq:wmr}f_q\(\mathbf x_t\):=\begin{cases}\phantom{-}1, &\text{if }\sum_{i=1}^nw_{i,t}x_{i,t}\ge \(2q-1\)w_t,\\ -1, & \text{otherwise}\end{cases}\end{equation}

for $q\in [0.5,1]$. If all votes are equally weighted, i.e., if $w_{i,t}=1$ for all $i=1,\dots,n, t>0$, then \eqref{eq:wmr} reduces to \eqref{eq:decision}.
\end{description}
Using this notation, our goal is to determine the weights $\vt w_t$ and the weighted majority rule $f_{\overline q_t}$ -- or equivalently the quota $\overline q_t$ --  that optimize the collective welfare $E_t$ given the selected committee $N_t$ of validators at time slot $t$, i.e., 

\begin{equation}\label{eq:lp}\max_{q, \vt w_t} {\left\{E_t\(f_q,w_t\)\right\}}
\end{equation}
This is the statement of the following Theorem which is due to \cite{Be97} and in special instances due to \cite{Ni82, Sh84}. 
\begin{theorem}[Optimal Weighted Voting Scheme,\cite{Be97}]\label{thm:main}
Consider a committee $N_t=\{1,2,\dots,n\}$ of validators with voting pro\-files $\vt p_t=\(p_{i,t},p_{2,t},\dots,p_{n,t}\)$ that have been selected to vote on the proposed block in time slot $t>0$. Then, given $\alpha$ and the collective welfare parameters $\ell_a,\ell_r$ in \Cref{tab:welfare}, the decision rule that maximizes the collective welfare, cf. \eqref{eq:lp}, is given by the weighted majority rule $f_{\overline q_t}$, with quota 
\begin{align}\label{eq:optimal}\overline q_{t}&:= \frac12\left[1-\(\ln{\(\frac{1-\alpha}{\alpha}\)}+\ln{\(\frac{1+\ell_r}{1+\ell_a}\)}\)w_t^{-1}\right]
\intertext{and individual weights}
\label{eq:weights} w_{i,t}&:=\ln{\(\frac{p_{i,t}}{1-p_{i,t}}\)}, \qquad \text{for }i=1,2,\dots,n.
\end{align}  
\end{theorem}

\begin{proof} To solve \eqref{eq:lp}, we will control the decision rule $f_q$ in \eqref{eq:welfare}. Then, the optimal quota and individual weights in \eqref{eq:optimal} and \eqref{eq:weights} will follow naturally. To see this, we expand \eqref{eq:welfare} using \eqref{eq:probability} which gives 

\begin{align*}
E_t\(f_q\)&=\(1-\alpha\)\(1+\ell_r\)\pi_{1}\(f_q\)+\alpha\(1+\ell_a\)\pi_{-1}\(f_q\)\\
&=\(1-\alpha\)\(1+\ell_r\)\sum_{\vt x_t:f_q\(\vt x_t\)=1}\(\prod_{i: x_{i,t}=1}p_{i,t}\prod_{j: x_{j,t}=-1}\(1-p_{j,t}\)\)+\alpha\(1+\ell_a\)\sum_{\vt x_t:f_q\(\vt x_t\)=-1}\(\prod_{i: x_{i,t}=-1}p_{i,t}\prod_{j: x_{j,t}=1}\(1-p_{j,t}\)\)\\
&=\sum_{\vt x_t:f_q\(\vt x_t\)=1}\(\(1-\alpha\)\(1+\ell_r\)\prod_{i: x_{i,t}=1}p_{i,t}\prod_{j: x_{j,t}=-1}\(1-p_{j,t}\)\)+\sum_{\vt x_t:f_q\(\vt x_t\)=-1}\(\alpha\(1+\ell_a\)\prod_{i: x_{i,t}=-1}p_{i,t}\prod_{j: x_{j,t}=1}\(1-p_{j,t}\)\)
\end{align*}
where we used that $p_{i,t}$ denotes the probability of a correct decision, cf. \eqref{eq:correct}. Our control variable to optimize the expected collective welfare in the previous equation is the decision rule $f_q$. It is immediate, that a general decision rule $f$ which maximizes $E_t$ is given by the following condition: for each decision profile $\vt x_t\in X=\{-1,1\}^n$ append every block for which 

\begin{equation}\label{eq:comparison}\(1-\alpha\)\(1+\ell_r\)\prod_{i: x_{i,t}=1}p_{i,t}\prod_{j: x_{j,t}=-1}\(1-p_{j,t}\)>\alpha\(1+\ell_a\)\prod_{i: x_{i,t}=-1}p_{i,t}\prod_{j: x_{j,t}=1}\(1-p_{j,t}\)\end{equation}
and reject it otherwise. The rule compares the contribution of the decision profile $\vt x_t$ to the expected collective welfare if it is accepted (left hand side) to its contribution if it is rejected (right hand side) and returns the outcome that yields the maximum contribution. This establishes its optimality. It remains to translate this rule into a form decision rule in the form of \eqref{eq:wmr} from which we will be able to deduct the optimal quota and weights. By grouping similar terms, we can rewrite \eqref{eq:comparison} as 

\begin{equation*}
\frac{1-\alpha}{\alpha}\cdot\frac{1+\ell_r}{1+\ell_a}\prod_{i: x_{i,t}=1}\frac{p_{i,t}}{\(1-p_{i,t}\)}>\prod_{j: x_{j,t}=-1}\frac{p_{i,t}}{\(1-p_{j,t}\)}
\end{equation*}

and by taking the natural logarithm
\begin{equation*}
\ln{\(\frac{1-\alpha}{\alpha}\)}+\ln{\(\frac{1+\ell_r}{1+\ell_a}\)}+\sum_{i: x_{i,t}=1} \ln{\(\frac{p_{i,t}}{1-p_{i,t}}\)}>\sum_{j: x_{j,t}=-1} \ln{\(\frac{p_{j,t}}{1-p_{j,t}}\)}
\end{equation*}
where we used that the function $\ln$ is monotone increasing. To proceed, we will use the transformations $y_i:=\(x_i+1\)/2$, so that 
\[y_i=\begin{cases}1, & \text{if } x_i=1,\\ 0, & \text{if } x_i=-1 \end{cases}\]
and $z_i:=\(1-x_i\)/2$, so that 
\[z_i=\begin{cases}1, & \text{if } x_i=-1,\\ 0, & \text{if } x_i=1 \end{cases}\]
Using this notation, we may rewrite the last inequality as 
\[\ln{\(\frac{1-\alpha}{\alpha}\)}+\ln{\(\frac{1+\ell_r}{1+\ell_a}\)}+\sum_{i=1}^n\left[\ln{\(\frac{p_{i,t}}{1-p_{i,t}}\)}y_i -\ln{\(\frac{p_{i,t}}{1-p_{i,t}}\)}z_i\right]>0\]
Now, since $y_i-z_i=x_i$, we obtain that 
\begin{equation}\label{eq:rule}\sum_{i=1}^n\ln{\(\frac{p_{i,t}}{1-p_{i,t}}\)}x_i>-\ln{\(\frac{1-\alpha}{\alpha}\)}-\ln{\(\frac{1+\ell_r}{1+\ell_a}\)}\end{equation}
where the left hand side equals the left hand side of \eqref{eq:wmr} with $w_{i,t}:=\ln{\(\frac{p_{i,t}}{1-p_{i,t}}\)}$. This establishes \eqref{eq:weights}. To obtain \eqref{eq:optimal}, it remains to bring the right hand side of this inequality in the form of \eqref{eq:wmr}. This is equivalent to determining $\overline q_t$ so that the following equality holds
\[\(2\overline q_t-1\)w_t=-\ln{\(\frac{1-\alpha}{\alpha}\)}-\ln{\(\frac{1+\ell_r}{1+\ell_a}\)}\]
Solving for $\overline q_t$ yields \eqref{eq:optimal}.
\end{proof}

\begin{remark}\label{rem:work} 
Based on the findings of \Cref{thm:main} about the optimal quota and weights, cf. cf. \eqref{eq:optimal} and \eqref{eq:weights}, we have the following remarks. 
\begin{itemize}[leftmargin=*, noitemsep]
\item The optimal quota \eqref{eq:optimal} depends on the validators' profiles and hence, it may vary between different time slots $t>0$. Also, it may vary according to the values of the parameters $\alpha, \ell_r,\ell_a$. For instance, the selection $\alpha=1/2$ neutralizes the bias due to assumptions on the distribution of valid versus invalid blocks whereas the selection $\ell_r=\ell_a$ neutralizes the bias due to differences in perceived network costs/rewards. In this way, \Cref{thm:main} maximizes the collective welfare -- or equivalently, the probability of consensus on the correct decision -- by an easily adaptable and \emph{dynamic} decision rule.\par
\item In blockchain applications, it may be desirable to enforce certain restrictions, as for example that the required weighted majority will be no less than $2/3$ of the total weights or that each individual weight will be no less than some threshold value. As \cite[p.332]{Sh84} explains, even in such cases of additional restrictions and/or perturbed assumptions, selecting weights that are proportional (or equal) to the optimal ones \eqref{eq:weights} will improve the probability of a correct outcome compared to unweighted decision making.\par
\item For any voting profile $p\in\(0,1\)$, the optimal weight $w=\ln{\(p/\(1-p\)\)}$ is given in \Cref{fig:work}.
\begin{figure}[!htb]
\centering
\includegraphics{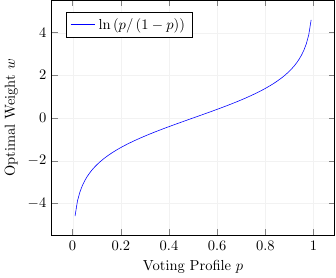}
\caption{Optimal weight $w=\ln{\(p/\(1-p\)\)}$ for any possible voting profile $p\in\(0,1\)$.}
\label{fig:work}
\end{figure}
For voting profiles $p<0.5$, the weight is negative which motivated our selection of suspending validators with profiles less that $0.5$. Intuitively, this means that these validators behave more than $50\%$ of the time not as expected and hence have an adverse impact to the consensus mechanism. Such validators could be even replaced by a random coin (which would be correct $50\%$ of the time). By properly initializing and suspending profiles above the $0.5$ threshold, such mathematically trivial situations can be avoided. This is discussed in more detail in \Cref{sec:implementation}.
\item On the other hand, the optimal weights increase steeply as the voting profiles approach $1$. For instance, a validator with voting profile $p_1=0.9$ has an optimal weight $w_1\approx2.2$, whereas a validator with voting profile $p_2=0.99$ has an optimal weight $w_2\approx4.6$. In concrete terms, this means that a $10\%$ improvement in the validator's voting profile resulted in a more than $100\%$ improvement in the validator's optimal weight. In this way, the weighting mechanism disproportionally punishes even the slightest deviations from the optimal behavior and hence, incentivizes validators to behave correctly as consistently as possible. This feature can be utilized to enforce the socially desired outcome of almost zero failures in the consensus mechanisms. 
\item Finally, note that for applications, \eqref{eq:wmr} can be simplified as follows. First, by dividing both sides with $w_t$, we can normalize the weights to sum up to $1$, i.e. $w'_{i,t}:=w_{i,t}/w_t$ for $i=1,2,\dots,n$. Second, by letting $y_i=\(x_i+1\)/2$ as in the proof of \Cref{thm:main}, with $y_i=1$ if $x_i=1$ and $y_i=0$ if $x_i=-1$, we can rewrite \eqref{eq:wmr} as 
\[\sum_{i=1}^n w'_{i,t}\(2y_i-1\)\ge 2\overline q_t-1\]
or equivalently as 
\[2\sum_{i=1}^n w'_{i,t}y_i-\sum_{i=1}^nw'_{i,t}\ge 2\overline q_t-1\]
Since $\sum_{i=1}^nw'_{i,t}=1$, dividing both sides by $2$ yields the approval condition 
\begin{equation}\label{eq:condition} \sum_{i=1}^n w'_{i,t}y_i\ge \overline q_t\end{equation}
which is more handy than (but equivalent to) \eqref{eq:wmr} and can be hence used for applications. This is done in \Cref{ex:continued} and \Cref{ex:obvious} below.
\end{itemize}
\end{remark}

\addtocounter{theorem}{-2}
\begin{example}[Continued]\label{ex:continued}
Assume for simplicity that $\alpha=1/2$ and that $\ell_r=\ell_a$. In this case, the optimal rule is simple weighted majority, i.e., $\overline q=1/2$, or by substituting in \eqref{eq:wmr}, $f_{\overline q}\(\vt x\)=1$ if $\sum_{i=1}^nw_{i}x_{i}\ge0$ (dependence on $t$ is omitted to simplify the notation). Using \eqref{eq:weights} and normalizing the weights to sum up to $1$, we obtain $\vt w=\(0.392,0.392,0.072,0.072,0.072\)$. With these choices, the probability of approving a valid block is approximately $0.927$ as shown in \Cref{tab:scale1}. \begin{table}[!htb]
\centering
\begin{tabular}{@{}llrrrrrrr@{}}
\midrule
\textbf{Profiles} & \textbf{Weights}&\multicolumn{7}{l@{}}{Decision profiles $\vt x$, with $f_{\overline q}\(\vt x\)=1$.}\\\midrule
0.9 & 0.392 & 1 & 1 & 1 & 1 & -1 & -1 & -1   \\
0.9 & 0.392 & 1 & -1 & -1 & -1 & 1 & 1 & 1  \\
0.6 & 0.072 & -1,1 & 1 & 1 & -1,1 & 1 & 1 & -1,1\\
0.6 & 0.072 & -1,1 & 1 & -1,1 & 1 & 1 & -1,1 & 1\\
0.6 & 0.072 & -1,1 & -1,1 & 1 & 1 & -1,1 & 1 & 1 \\\midrule
\multicolumn{2}{@{}r}{$\sum_{i=1}^5{w_ix_i}\phantom{22}$} & $>0$ & $>0$ & $>0$ & $>0$ & $>0$ & $>0$ &$>0$\\\midrule  
\end{tabular} 
\caption{Winning profiles under the optimal decision scheme.} 
\label{tab:scale1}
\end{table}
The weighting has resulted in a voting rule which approves a block if the two high-profile (0.9) validators agree on its validity (first column of decision profiles). The votes of the remaining va\-lidators come into play only if these two disagree. In this case, it suffices that a majority ($2$ out of $3$) of the remaining validators approve the block (remaining $6$ columns of decision profiles). The probabilities of decision profiles $\vt x_t$ with $f_{\overline q_t}\(\vt x_t\)=1$ sum up to approximately $0.927$ as claimed

\[0.9^2+\dbinom{2}{1}0.9\cdot0.1\(\dbinom{3}{3}0.6^3+\dbinom{3}{2}0.6^2\cdot0.4\)\approx 0.927\]
The weights $\vt v_t=\(1/3,1/3,1/9,1/9,1/9\)$ yield the same result and are in that sense, equivalent to the optimal ones. In fact, there may be several other optimal choices. As mentioned above, if we impose additional restrictions such as a de facto $2/3$-weighted majority, the weights given by \eqref{eq:weights} may not be optimal anymore. However, as \cite{Sh84} remarks, they will still yield an improved probability compared to the unscaled case. In this example, a similar calculation as in \Cref{tab:scale1} shows that the probability of reaching the $2/3$-majority is $0.81$ with the $\vt w_t$ weights and approximately $0.85$ with the $\vt v_t$ weights.
\end{example}
\addtocounter{theorem}{+2} 

The \textit{ProcessSlot} procedure in \Cref{alg:wvoting} executes the weighted voting algorithm in each slot $t$. First, the validators and a block proposer are selected according to the underlying PoS protocol (as part of the \textit{PoSGetCommittee} procedure). The committee can be a random sample or deterministic subset taken from the entire validator set, or it can be the validator set in its entirety. For example, in Delegated Proof of Stake as implemented in TRON and EOS.IO, \textit{PoSGetCommittee} would return the set of superdelegates. Next, the block proposed by the block proposer is retrieved and stored in the variable $B$ -- if no valid block has been proposed, then we assume that the value \textit{null} is stored in $B$. Once the committee has been formed and a valid block has been proposed, we retrieve for each committee member the block that they voted for using the \textit{CollectVotes} procedure. We then check in \textit{ProcessVotes} whether the validators voted for $B$ -- if so, we record their vote as a $1$, and as a $-1$ otherwise. Next, the weighted voting procedure\footnote{The function that determines the optimal quota is given in a general form (line 26) to account for implementations with a rule different from the one proposed in \eqref{eq:optimal}, as e.g., a constant $2/3$-majority rule.} is applied independent of the underlying PoS protocol -- if a sufficient number of committee members voted for the $B$, then it is committed and appended to the blockchain. In this way, it contributes towards a more efficient and fair consensus mechanism while remaining decoupled from the PoS selection mechanism. This results in a two-layered scheme that on the one hand improves the efficiency of the consensus mechanism of the PoS protocol and on the other hand can be implemented and reverted with minimal cost to the users.

\newcommand{\vote}{\textit{vote}}
\newcommand{\B}{\mathcal{B}}
\begin{algorithm}[!hbt]
\caption{Weighted Voting in Committees}\label{alg:wvoting}
\begin{algorithmic}[1]
\Procedure{ProcessSlot}{$t,(p(i))_{i \in I}$}
\State $(v_1,\ldots,v_n) \gets $ \Call{PoSGetCommittee}{$t$}
\State $B \gets $ \Call{PoSProposeBlock}{$(v_1,\ldots,v_n), t$}
	\If{$B \neq \textit{null}$} 
	\State $(x_1,\ldots,x_n) \gets$ \Call{CollectVotes}{$(v_1,\ldots,v_n),B$}
	\State $(\vote(1),\ldots,\vote(n)) \gets$ \Call{ProcessVotes}{$t$}
	\If{\Call{WeightedVoting}{$(p(v_1),\ldots,p(v_n), (\vote(1), \ldots,\vote(n))$}}
		\State \Call{CommitBlock}{B}
	\EndIf
\EndIf
\EndProcedure
\Procedure{ProcessVotes}{$(x_1,\ldots,x_n),B$}
\For {$i \in \{1,\ldots,n\}$} 
\If {$x_i==B$} $\Comment{x_i=i\text{'s vote message}}$ \State {$\vote(i) \gets 1$}
\Else \State{$\vote(i) \gets -1$}
\EndIf
\EndFor
\EndProcedure
\Procedure{WeightedVoting}{$\(p_{i}\)_{i\in N}, \(\text{vote}\(i\)\)_{i\in N}$}
\For {$i \in \{1,\ldots,n\}$}
\State {$w_i \gets \ln{\(p_i\)}-\ln{\(1-p_i\)}$} \label{ln: logs}
\State {sum $\gets$ sum $+w_i\cdot$vote$\(i\)$}
\EndFor
\State $q=\Call{OptimalQuota}{\(w_i\)_{i\in N},\alpha,\ell_r,\ell_a}$
\Return {sum $\ge 2q-1$}
\EndProcedure
\end{algorithmic}
\end{algorithm}

\section{Multiplicative Weights Update Algorithm}\label{sec:mwua}
We now turn to one of the main challenges of implementing the voting profile scheme in the dynamic blockchain setting, which is the \emph{update} of voting profiles after every time slot. The updating scheme may considerably vary depending on the desired result: e.g., in \cite{Yu18}, a reputation system is proposed in which reputation increases according to a \emph{sigmoid function} when nodes vote correctly and decreases sharply (to $0$) after a single violation. While this approach adheres to intuition and comes with certain merits, practical applications may call for more flexibility in the updates. \par
To develop a parametrizable scheme, we utilize the approach of \cite{Ar12} who generalize the standard \emph{multiplicative weights update} (MWU) algorithm to a non-binary setting in which experts' scores are revised according to the impact of their decision on the social welfare. Using \Cref{tab:welfare}, the corresponding MWU algorithm for the present application is given in \Cref{tab:updates}.
\begin{table}[!htb]
\centering
\begin{tabular}{@{}lrll@{}}
& & \multicolumn{2}{c}{\textbf{Proposed Block $B_t$}}\\[0.2cm]
&& Valid (1) & Invalid (-1)\\
\cmidrule{2-4}
\multirow{2}{*}{\textbf{Committee\hspace{4pt}}}& Approve & $p_{i,t}\(1+\delta\)$ & $p_{i,t}\(1-\delta\)^{\ell_a}$\\
\cmidrule{2-4}
& Reject & $p_{i,t}\(1-\delta\)^{\ell_r}$ & $p_{i,t}\(1+\delta\)$\\
\cmidrule{2-4}
\end{tabular}
\caption{Multiplicative Weights Updates.}
\label{tab:updates}
\end{table}
Here, $\delta>0$ is a (small) number subject to the exact protocol parametrization. Apart from the efficiency properties of the MWU algorithm that are well known, see \cite{Ar12,Pa17,Ba18} and references therein, this scheme can leverage the prevailing network conditions and adjust the updates accordingly. This can be realized by replacing $\delta$ and/or $\ell_r,\ell_a$ with \emph{sequences of updating parameters}, $\(\delta_t, \ell_{r,t},\ell_{a,t}\)_{t>0}$. The proposed scheme for updating validators' voting profiles is given for completeness in \Cref{alg:profiles}. If a new node comes online ($T$ is the slot when it does), it runs the \textit{MWInitialize} procedure to get the voting profiles consistent with the rest of the network. Once this has been completed, the function \textit{MWUpdate} is executed every slot to update the voting profiles. If a validator's profile drops below 0.5, they are removed from the set -- this happens in line~\ref{ln:suspend}. Similarly, new validators may be added in each round, depending on the protocol implementation. This processed in line~\ref{ln:add}. The \emph{min} in line~\ref{ln:epsilon} ensures that the voting strength of a single validator cannot go to infinity (after all, if $p(i)$ approaches 1 then $w_i$ approaches $\infty$ as per line~\ref{ln: logs} in \Cref{alg:wvoting}).

\begin{algorithm}[!htb]
\caption{Validators' Voting Profiles}\label{alg:profiles}
\begin{algorithmic}[1]
\Procedure{MWInitialize}{$I_0,\delta,\ell_r,\ell_a, T$}
\State $I \gets I_{0}$
\State $t \gets 1$
\For{$i \in I$}
	$p(i) \gets 0.5$
\EndFor
\While{$t \leq T$}
	\State $(I, (p(i))_{i \in I}) \gets$ \Call{MWUpdate}{$I,(p(i))_{i \in I},t,\delta,\ell_r,\ell_a$}
	\State $t \gets t+1$
\EndWhile
\EndProcedure
\Procedure{MWUpdate}{$I,(p(i))_{i \in I},t,\delta,\ell_r,\ell_a$}
\State $(v_1,\ldots,v_n) \gets $ \Call{PoSGetCommittee}{$t$}
\State $B \gets $ \Call{PoSProposeBlock}{$(v_1,\ldots,v_n), t$}
\State $(x_1,\ldots,x_n) \gets$ \Call{CollectVotes}{$(v_1,\ldots,v_n),B$}
\For{$i \in I$}
	\If {$v_i=B$}
		\State {$p(v_i) \gets \min{\{1-\epsilon,p(v_i) \(1+\delta\)\}}$} $ 	
		\Comment{\epsilon=10^{-5}}$ 
		\label{ln:epsilon}
	\ElsIf {$x_i = $ \textit{null}}
		\State {$p(v_i) \gets p(v_i) \(1-\delta\)^{\ell_r}$}
	\Else 
		\State {$p(v_i) \gets p(v_i) \(1-\delta\)^{\ell_a}$}  
	\EndIf
\EndFor
\If {$p(v_i)< 0.5$ and $t\ge g$} $\Comment{\text{$g$ is the grace period }}$
\State $I \gets I \backslash \{v_i\}$ $\Comment{\text{suspend validator $i$}}$ \label{ln:suspend}
\EndIf
\State $(u_1,\ldots,u_{k}) \gets$ \Call{GetNewValidatorsInSlot}{$t$}
\For{$i \in \{1,\ldots,k\}$}
	\State $I \gets I \cup \{u_i\}$ \label{ln:add}
	\State $p(u_i) \gets 0.5$
\EndFor
\Return $(I, (p(i))_{i \in I})$
\EndProcedure
\end{algorithmic}\vspace{0cm}
\end{algorithm} 

\subsection{Tolerance of Validator Faulty Behavior}\label{sec:faulty}
Before testing the weighted voting and MWU algorithms numerically in the \Cref{sub:numerical}, we study analytically the effect of the updating scheme on the validators' profiles. Due to the parametric nature of the proposed algorithm, the question that naturally arises from the validators' perspective is the following: \emph{what is the threshold of faulty behavior that can be sustained by the algorithm without harming my profile?} In other words, what is the \emph{minimum} percentage of time that a validator should strive to behave properly in order to sustain or improve their voting profile. \par
To explore this question, we first introduce some notation. Let $q\in[0,1]$ denote the percentage of time -- measured in time slots, cf. \Cref{sec:model} -- that a (tagged) validator behaves according to the protocol, i.e., votes as expected by default, and let $q_1\in[0,1]$ denote the percentage of time that this validator does not vote on valid blocks (e.g., because they accidentally drop offline when it is their turn to vote). Accordingly, it must be that $q+q_1\le 1$, with $1-q-q_1\in[0,1]$ denoting the percentage of time that the validator votes on the approval of invalid blocks. Finally, let $p_0$ denote the starting profile of the validator at the period under scrutiny and for any time frame $T>0$, let $p_T$ denote the validator's profile at time $T$. Our aim is to obtain necessary and sufficient conditions on $q$ and $q_1$, so that the validator will improve or sustain their initial voting profile, i.e., $p_T\ge p_0$, where $p_T$ depends on the validator's voting behavior via $q$ and $q_1$. Typically, we will be interested in large time periods, $T\gg0$, that properly reflect the validator's voting behavior as expressed by $q$ and $q_1$. However, the result of \Cref{thm:trivial} remains correct for any $T>0$.

\begin{theorem}\label{thm:trivial} Consider the updating scheme in \Cref{tab:updates} with parameters $\delta,\ell_r,\ell_a>0$ such that $\ell_r\ll \ell_a$ and $\delta\in\(0,1\)$, (typically $\delta$ is a small number close to $0$). A validator who follows the protocol $q\cdot100\%$ of the time and does not vote on valid blocks $q_1\cdot 100\%$ of the time, with $q,q_1\in[0,1]$ and $q+q_1\le 1$, will improve upon their initial profile if and only if 
\begin{equation}\label{eq:trivial} q\ge c_1\(1-c_2q_1\)\end{equation}
where $c_1:=\(1-\frac{\ln{\(1+\delta\)}}{\ell_a\ln{\(1-\delta\)}}\)^{-1}$ and $c_2:=1-\ell_r/\ell_a$ are positive constants that depend on the updating parameters $\delta,\ell_r$ and $\ell_a$.  
\end{theorem}
\begin{proof} The positivity of $c_1$ and $c_2$ follows directly from the assumptions. Indeed, $\ell_r<\ell_a$ implies that $\ell_r/\ell_a<1$ and hence, that $c_2>0$. Similarly, $\delta\in\(0,1\)$ implies that $\ln{\(1-\delta\)}<0$ and $\ln{\(1+\delta\)}>0$ and hence, that $c_1>0$ as the inverse of the sum of two positive terms (in particular, $c_1<1$). \par
Concerning \eqref{eq:trivial}, let $p_0$ denote the validators initial voting profile and let $p_T$ denote the validator's voting profile after $T$ time slots. Then, given that $q$ denotes the fraction of time slots in which the validator voted correctly and $q_1$ the fraction of time slots in which the validator did not approve a valid block, we have that 
\begin{align*}p_T&=\(1+\delta\)^{qT}\(\(1-\delta\)^{\ell_r}\)^{q_1T}\(\(1-\delta\)^{\ell_a}\)^{\(1-q-q_1\)T}p_0=\(\(1+\delta\)^q\(1-\delta\)^{\(\ell_rq_1+\ell_a\(1-q-q_1\)\)}\)^Tp_0\end{align*}
Hence, $p_T\ge p_0$ if and only if $\(\(1+\delta\)^q\(1-\delta\)^{\(\ell_rq_1+\ell_a\(1-q-q_1\)\)}\)^T\ge1$. Since, the basis term is positive by assumption, we may take the natural logarithm of both sides to obtain the equivalent condition 
\[q\ln{\(1+\delta\)}+\(\ell_rq_1+\ell_a\(1-q-q_1\)\)\ln{\(1-\delta\)}>0\]
Using that $\(1+\delta\)>1$ and some standard algebraic operations, we can solve the last inequality for $q$ to obtain the necessary and sufficient condition
\[q\ge \frac1{1-\frac{\ln{\(1+\delta\)}}{\ell_a\ln{\(1-\delta\)}}}\cdot\(1-q_1\(1-\ell_r/\ell_a\)\)\] 
which concludes the proof.
\end{proof}

\begin{remark}[Policy Implications]\label{rem:trivial}
The result of \Cref{thm:trivial} can be utilized in the design policy of the updating algorithm to enforce certain requirements in the validators' realized correct behavior, $q$. To see this, we first observe that $c_1$ can be rewritten as 
\[c_1=\ell_a\cdot\(\ell_a-\frac{\ln{\(1+\delta\)}}{\ln{\(1-\delta\)}}\)^{-1}\]
Since $\delta$ is typically a small positive number, e.g. $\delta=10^{-3}$ and $\ell_a$ is a positive number large enough to represent the potential losses from appending an invalid block to the blockchain, $c_1$ gets arbitrary close to, but strictly less than $1$ for increasing values of $\ell_a$. Moreover, for large values of $\ell_a$ and small values of $\ell_r$, $c_2:=1-\ell_r/\ell_a$ is also close to, but strictly less than $1$. This implies that $1-\varepsilon<c_1c_2<1$ for some relatively small $\varepsilon>0$. Hence, using that $q+q_1\le 1$, we may rewrite \eqref{eq:trivial} as

\begin{equation*}
q+c_1c_2q_q\ge c_q \iff \(1-c_1c_2\)q+c_1c_2\(q+q_1\)\ge c_1 \implies \(1-c_1c_2\)q+c_1c_2\ge c_1\implies q\ge \frac{c_1-c_1c_2}{1-c_1c_2}
\end{equation*}
This gives a lower bound on $q$ that does not depend on $q_1$. It is obvious that if the parameters $\delta, \ell_r,\ell_a$ are selected in a way to push $c_1$ close to $1$, then a validator will have to vote consistently correctly to retain or improve their voting profile. In general, this shows the advantages of the parametric approach and the flexibility that it conveys to the policy makers or consensus designers (in case of public, permissionless blockchains) to enforce desired outcomes or even to dynamically adapt the consensus algorithm in response to prevailing network conditions. 
\end{remark}

\section{Numerical Tests \& Practical Use Cases}
\label{sub:numerical}
To visualize the proposed scheme, we study some instantiations of the weighted voting and MWU algorithms. Before going into the details of each case, the following hold in general.
\begin{itemize}[leftmargin=*]
\item \textbf{Adversarial model:} an adversary blocks $v\%$ of the votes, either by abstaining (accidentally or intentionally) or by censoring other validators. To demonstrate the efficiency of the model in extreme -- or worst case -- conditions, we show simulations for $v=40\%, 50\%$ and $60\%$. For lower values of $v$, i.e., for less adversarial power, the results are similar (better) and thus omitted.
\item \textbf{Parameter choice:} $\alpha$, the prior probability that a proposed block is invalid, is set to $1/2$. This choice neutralizes the bias due to assumptions on the distribution of valid-invalid blocks in \eqref{eq:optimal} and corresponds to the most general model. To capture that costs from approval of invalid blocks are higher than costs from rejection of valid blocks, we select $\ell_a=12 \gg 10^{-2}=\ell_r$. The resulting graphs (recovery times) are robust in a wide range of these parameters. Yet, very low values of $\ell_a$ may lead to unwanted behavior: validators who are commonly off-line may still improve their profiles despite not voting on valid blocks. Finally, the updating parameter $\delta$ has been chosen according to the related literature \cite{Ar12, Ba18}. Different values of $\delta$ affect the rate of convergence of the profiles.
\item \textbf{Simulation environment:} The numerical results have been established in MATLAB R2018b. The simulations validate the algorithms for resuming consensus, cf. \Cref{fig:sim_60,fig:sim_50} and for updating a validator's profile, cf. \Cref{fig:evolution_1,fig:evolution_2}. Further issues related to scalability and computational complexity on a proper -- or simulated -- blockchain are discussed in \Cref{sec:implementation}. 
\end{itemize} 
\Cref{fig:sim_60} illustrates a scenario with an adversary blocking $40\%$ of the votes. At the start of the attack, all $n$ validators\footnote{The exact number does not change the results. For simplicity, we used $n=100$.} have a voting profile $p=0.9$. We examine two choices of the updating parameter $\delta=10^{-3}$ (red) and $\delta=2\times 10^{-3}$ (blue). In both cases, $\alpha=1/2$ and $\ell_r=10^{-2}, \ell_a=12$.
\begin{figure}[!htb]\vspace{-0.2cm}
\centering
\includegraphics{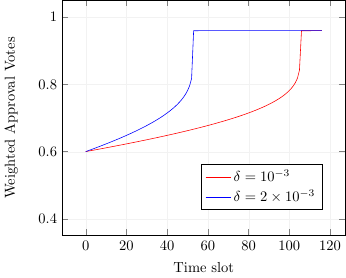}
\caption{Time slots to resume consensus on valid blocks with $40\%$ non-voting nodes under mild (red line) and aggressive (blue line) updating parameter $\delta$. In both cases, $\ell_r=10^{-2}$.}
\label{fig:sim_60}\vspace*{-0.2cm}
\end{figure} The depicted curves indicate that the weighted approval votes (vertical axis) rise above the $2/3$ majority threshold\footnote{While the optimal quota, cf. \eqref{eq:optimal}, remains slightly above $1/2$, we consider the $2/3$-majority rule which is currently implemented in PoS protocols.} for both cases. The pace is different and depends on the selection of $\delta$. The sharp bends in both lines correspond to the point in which the scores of voting validators numerically reach $1$. After this point, the majority of the voting validators increases at a very slow pace which is a desirable property that allows for a more smooth recovery in case that the abstaining $40\%$ resume voting. Specifically, the exact formula for updating their profiles is 
\[p_{i,t+1}=\min{\{1-10^{-15},\max{\{0.5,p''_{i,t+1}\}}\}}\]
where $p''_{i,t+1}$ is the value calculated by \Cref{tab:updates}. This formula ensures that $p_{i,t}$ remains in $[0.5,1)$. 
\begin{figure}[!hbt]
\centering
\includegraphics{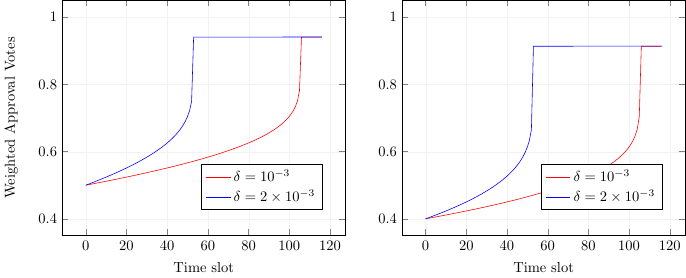}
\caption{Time slots to resume consensus on valid blocks with $50\%$ (left panel) and $60\%$ (right panel) non-voting nodes under mild (red line) and aggressive (blue line) updating parameter $\delta$. In both cases, $\ell_r=10^{-2}$.}
\label{fig:sim_50}\vspace*{-0.1cm}
\end{figure}
\begin{remark}
To avoid the numerical instability at $p=1$ -- for which the denominator at $\ln{\(1/\(1-p\)\)}$ becomes equal to $0$ -- we need to restrict profiles away from $1$. Specifically, to implement and test this scheme, we require that voting profiles do not exceed $1-\epsilon$ for some very small $\epsilon>0$. However, due to the steep increase of the optimal weights that was mentioned in \Cref{rem:work}, the actual value of the imposed threshold $\epsilon$ matters and it should be chosen to be numerically close to $0$. Moreover, again due to the disproportional increase of the optimal weights for profiles close to $1-\epsilon$, validators who achieve this threshold (and remain at it by continuing to vote correctly after reaching it), have disproportionally higher power to influence the consensus outcome.
\end{remark}
As a comparison, \Cref{fig:sim_50} illustrates the process of resuming approval of blocks in the same scenario but with an adversary controlling $50\%$ of the stake (left panel) and $60\%$ of the stake (right panel). The results indicate a very similar recovery pattern, cf. \Cref{fig:sim_60}, independently of the initial stake of the non-voting validators. The evolution of a validator's voting profile is illustrated in \Cref{fig:evolution_1}. In the depicted scenario, the validator's initial profile is $0.9$. The validator votes correctly $80\%$ of the time, but drops $10\%$ of the time off-line, and votes on invalid blocks another $10\%$ of the time. Again, we examine two cases for different values of the update parameter, $\delta=2\times 10^{-2}$ (left panel) and $\delta=10^{-2}$ (right panel).
\begin{figure}[!htb]
\centering
\includegraphics{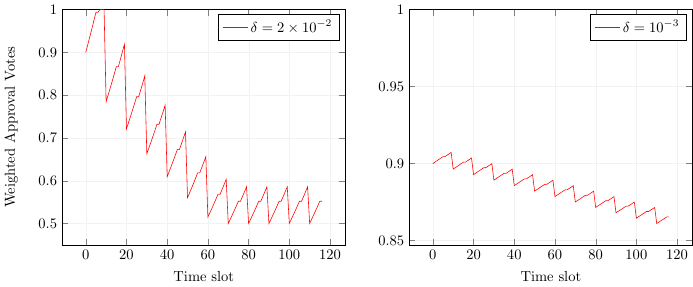}
\caption{Evolution of a validator's voting profile who periodically drops offline and periodically votes on an invalid block for different values of parameter $\delta$. In both panels, the validator's initial score is $0.9$, $\ell_a=12$ and $\ell_r=10^{-2}$.}
\label{fig:evolution_1}\vspace*{-0.2cm}
\end{figure}
While the patterns differ, in both depicted panels the voting profile falls due to the regular incorrect votes. We remark, that lower values of $\ell_a$ would result in upwards sloping curves (not depicted here) implying that a validator could regularly vote incorrectly and still improve their voting profile. Similarly, higher values of $\delta$ would allow validators to quickly recover their profiles after pitfalls which is an undesirable property. The depicted patterns in the evolution of the voting profile are robust in the choice of the initial score and the value of $\ell_r$.\par 
Finally, \Cref{fig:evolution_2} extends the above scenarios to a period in which the validator resumes proper voting. Specifically, we assume that the validator votes correctly on every block except of occasional time slots -- less than $10\%$ of the time -- in which they go offline. Again, the two panels correspond to different values of $\delta$.
\begin{figure}[!htb]
\centering
\includegraphics{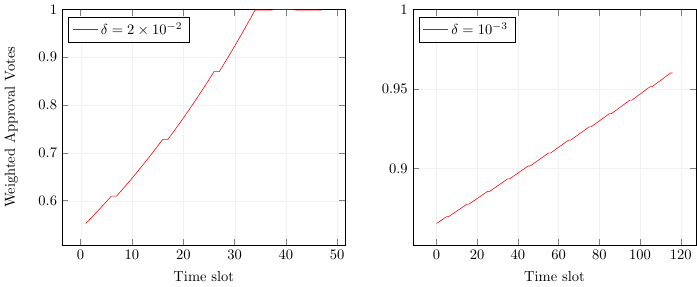}
\caption{Recovery of the validator's voting profile after resuming proper voting (with only occasional offline-drops) in the scenarios of \Cref{fig:evolution_1}. Recovery exhibits linear pattern for both choices of $\delta$. It is fast for $\delta=2\times 10^{-2}$ (aggresive adjustment) and slow for $\delta=10^{-3}$ (mild adjustments).}
\label{fig:evolution_2}\vspace*{-0.1cm}
\end{figure}
In both cases, the pattern is linear (the sharp bends correspond to the occasional drops) with a slope that can be adjusted by the choice of $\delta$. In sum, the simulations support the versatility of the proposed scheme and leave the exact parametrization (static or dynamic) subject to each protocol's implementation and scope.

\subsection{Use Case: Ethereum Blockchain with Casper FFG}\label{sub:use1}
To further test the proposed scheme in a potential practical scenario, we consider the implementation of weighted voting in Ethereum's Casper the Friendly Finality Gadget (FFG) as a PoS overlay on top of the PoW main chain\cite{Bu19}. This is a hybrid consensus system in which, roughly, PoS validators -- stakers -- vote to confirm (validate) blocks that have been proposed by \qt{conventional} PoW miners at regularly-occurring \emph{checkpoints}. To become a validator, a user has to \emph{deposit} a chosen amount of ETH, Ethereum's native cryptocurrency. If more than two thirds of the validators (where each validator is weighted proportionally to the size of their deposit) vote to approve a given checkpoint block, then it is considered \emph{justified}, and if two checkpoint blocks in a row are justified then the first block is considered \emph{finalized}. Ethereum nodes will never drop finalized blocks during a chain reorganization, so finalization provides an additional layer of security to the users. 
If more than one third of the validators go offline or are subjected to a network partition or eclipse attack, then checkpoints cannot be finalized. The system eventually recovers in the following way: after each checkpoint, those validators who did not vote are \emph{penalized} which means that their fraction of the total deposit decreases. Eventually, their deposit fraction will fall below one third -- hence, the voting validators regain a two-thirds supermajority and the ability to finalize. Of course, it would be harsh on the non-voting validators if their deposits were shrunk quickly during a network partition on eclipse attack, which would by itself disincentivize potential validators. The penalties are therefore initially very mild, but they become harsher over time. This is illustrated in \Cref{fig:ether}.
\par
\begin{figure}[!htb]
\centering
\includegraphics{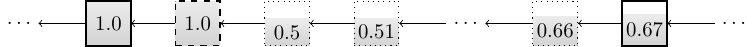}\\[0.1cm]
\caption{Illustration of Casper FFG. The boxes represent checkpoint blocks, and each value inside a box represents the fraction of validators that voted for that checkpoint. Thick edges around a box indicate that the checkpoint is finalized, dashed edges indicate mere justification, and a dotted edge a checkpoint that is neither.  If a considerable fraction (in this case $50\%$) of validators suddenly go offline, then finalization stalls. In Casper FFG, the protocol is eventually recovered by gradually slashing the deposits of offline validators. Using weighted voting, this can be achieved without causing financial losses to the validators.
For more details of the liveness of the Hybrid Casper proposal see \cite{Bu19}.}
\label{fig:ether}
\end{figure}
The main contribution of the currently proposed weighted voting scheme in Ethereum's Hybrid Casper Protocol setting is that it enables the protocol to recover without hurting the deposits of the validators. Instead, validators who come back online eventually regain their voting profile without any financial losses. This also allows for faster recovery during network partitions: although network partitions that last for more than a day are extremely rare, a quick recovery that is achieved by slashing the deposits of offline validators may discourage potential validators from depositing. 
\par
The practical advantages of weighted voting are achieved without affecting either the underlying PoW blockchain or the PoS checkpoint finalization protocol. Instead, it operates within the selected committees and hence affects in a minimal way any suggested/existing scheme that is functional on the blockchain. This means that this scheme can be launched and reverted with minimal impact on client implementations and updates.

\subsection{Use Case: Incentives and Risk Mitigation in Proof of Stake Protocols}\label{sub:use2}
Several current PoS proposals including Ouroboros \cite{ouro17,Bru18}, Casper FFG \cite{Bu1811,Bu1810}, and Delegated Proof of Stake (DPOS) as implemented in EOS.IO, suggest the following properties.
\begin{description}[leftmargin=*,noitemsep]
\item[P1:] a limited number of pools or committee members, which is achieved through a minimum deposit for each potential validator.
\item[P2:] a fixed period (number of protocol epochs) of time for which the staked deposit will remain locked. 
\end{description}
As we discuss in the following, weighted voting and updating schemes become particularly relevant under these conditions. Up to now, we have assumed that the PoS reward mechanism is retained, i.e., that validators are rewarded proportionally to their stake. However, if there are large discrepancies between the voting profiles, validators with high voting profiles contribute more to the consensus than validators with lower profiles. This may raise issues about the distribution of staking rewards and the adopted reward scheme. \par 
To deal with such issues in a similar setting of decentralized multi-agent decision-making, \cite{Ba08a} and \cite{Az09} study incentives under the conventional reward scheme that gives to the participating agents their Shapley \cite{Sh53,Sh54} or Banzhaf values \cite{Ba65,Du79}. In blockchain applications, the relevant question to address is whether the decision scheme together with the reward scheme create incentives for validators to merge, split or annex their deposits to larger pools\cite{Az09}. \cite{Ba08a} quantify the potential profit from creating various identities to participate in the weighted voting mechanism and show that given the number $n$ of participating agents, a Sybil attack cannot earn to a potential adversary more than $2n/\(n+1\)\approx 2$ times their initial values. Moreover, both \cite{Ba08a} and \cite{Az09} demonstrate that it is computational difficult, i.e., NP-hard, to find such profitable manipulations (mergers or splitters). \par
These findings are promising and suggest that reward schemes that have been developed in the context of traditional coalitional games can be used to incentivize proper behavior in the consensus mechanisms of blockchain protocols. By using properties P1.-P2. as above, we significantly mitigate the related risks. Indeed, setting a number of allowed validating pools (P1) together with a large enough minimum required deposit increase the difficulty for splitters. Yet, taking these precautions to the other extreme takes its toll on decentralization and hence such measurements should be exercised to a limited extent. In any case, the fixed period for which deposits remain locked (P2), seems to significantly reduce the potential for movements in the short run at no significant tradeoffs with other blockchain protocol design principles.

\subsubsection*{Eclipse Attacks}
The weighted voting scheme is engineered to accelerate recovery of consensus and prevent the blockchain from stalling when a fraction of validators is off-line or in general does not vote. Validators that continue to vote increase their voting profiles and hence, their power to influence the consensus outcome, whereas the voting profiles of non-voting validators deteriorate. Furthermore each validator's optimal weight depends only on that validator's voting behavior (or profile) and hence, can be computed independently of the committee in which this validator participates. However, this is not true for the \emph{relative power} of each validator in each committee. An example is given below. 
\begin{example} \label{ex:obvious}Consider a validator $i=1$ with profile $p_{1,t}=0.99$ and two different committees:
\begin{description}[leftmargin=*,noitemsep]
\item[Committee in $t=1$:] $n=10$ validators with voting profiles, $\vt p_1=\(0.99, 0.7,\dots, 0.7\)$, 
\item[Committee in $t=2$:] $n=10$ validators with voting profiles, $\vt p_2=\(0.99, 0.95,\dots, 0.95\)$.
\end{description}
As in the numerical tests, we will assume that $\ell_r=10^{-2},\ell_a=12$ and $\alpha=1/2$ for both committees. Applying \eqref{eq:optimal} on these values yields the optimal quotas $\overline q_1=0.604$ for $t=1$ and $\overline q_2=0.54$ for $t=2$. The unnormalized optimal weight of validator $1$ is equal to $w_1=\ln{\(0.99/\(1-0.99\)\)}=4.595$ in both committees. However, their normalized weights $w'_{i,t}, t=1,2$ which will be used to assess the block approval condition \eqref{eq:condition}, are equal to $w'_{1,1}=0.376$ and $w'_{1,2}=0.148$ respectively. This shows that validator $1$ is better off in the first committee. \par
To reach this conclusion, the second committee was selected with much worse -- in terms of their voting profiles -- validators. However, similar discrepancies in validator $1$'s optimal normalized weight can be observed even in the presence of only $2$ equally good validators.
\begin{description}[leftmargin=*,noitemsep]
\item[Committee in $t=3$:] $n=10$ validators with voting profiles, $\vt p_1=\(0.99, 0.99, 0.95, 0.7,\dots, 0.7\)$.
\end{description}
In this case, the optimal quota $\overline q_3$ is equal to $\overline q_3=0.57$ and validator $1$'s optimal normalized weight is equal to $w'_{1,3}=0.254$. Similarly, validator $2$'s optimal normalized weight $w'_{2,3}$ is also equal to $0.254$ and validator $3$'s optimal normalized weight is equal to $0.163$. This highlights the power of these $3$ validators to decide the consensus outcome in committee $3$.
\end{example}
The previous example reveals two interesting facts. On the one hand, the dynamic adjustments in the optimal quota and validator's normalized weights create a consensus mechanism with intriguing properties when compared to static alternatives. The versatility of this algorithm under the concurrent theoretical guarantees that it optimizes the probability of a correct decision, or more accurately the expected collective welfare, motivate further research both from theoretical and practical perspectives concerning its adoption.\par
On the other hand, the dependence of the validators' \emph{relative power} on the profiles of the other committee members creates a incentive for malicious behavior that should be taken into account. Specifically, since validators' relative power increases as the profiles of the other committee members decrease, this motivates adversarial validators to perform \emph{eclipse attacks} and prevent other validators from voting properly. However, for such attacks to be profitable, the attackers needs to sustain their profile which means that they need to keep voting properly. Moreover, the actual outcome on the overall collective welfare will depend on the trade-off between the size of the attacker's stake and the size of the eclipsed validator's stake along with detection and potential penalties on such behaviors. In any case, fixing parameters as in P1-P3 seems to reduce the arsenal of potential attackers but a precise verdict on whether such attacks can be profitable depends on the exact protocol parametrization and should be examined separately in every case of adoption.

\section{Design \& Limitations}\label{sec:implementation}

The introduction of validators' voting profiles and the improvement in the consensus mechanism come with a trade-off in terms of security. Since the system becomes dependent on information that can be retrieved from the blockchain -- validators' votes are stored as messages \cite{Bu1811} -- this raises new risks on adversarial manipulation of this information. In the current section, we try to address these risks alongside implementation issues and limitations. 

\subsubsection*{Defense against known attacks}
To defend against consensus level attacks, current PoS proposals leverage the fact that non-voting nodes can be identified and penalized \cite{Bu1810,fanto18}. Yet, these actions are ineffectual against adversaries who can censor other validators or replenish their own stake and withstand the penalties to retain more than the required consensus-quota of the total stake \cite{Gi17,Pa18}. Although pessimistic, the scenario in which adversaries sustain an attack despite suffering losses on their own stake gains credence in the presence of potential \emph{out-of-protocol} profits. Further, consensus mechanisms that rely on PoS selection are vulnerable to \emph{flash} or \emph{blindsiding} attacks conducted by entering nodes \cite{Bo15,Ch18} or to accidental faults such as network latency, bad connectivity or simple negligence. Weighted voting provides lines of defense (in an obvious way) against these kinds of attacks or faults while retaining the benefits of the underlying PoS design. In addition, the preserved reliance on PoS for the selection of validators in committees and the allocation of rewards, protects against adversarial nodes that maintain small stake but high voting profile or vice versa.

\subsubsection*{Valid-invalid blocks}
\label{sub:valid}
A likely contentious assumption of the present model is that proposed blocks can be identified as valid or invalid\footnote{This is a \href{https://en.wikipedia.org/wiki/Chicken_or_the_egg}{chicken-egg} problem: if we can identify the \qt{canonical} blocks, then we can improve consensus with a scheme as the one proposed here. But in blockchains, the \qt{canonical} blocks are precisely the ones for which consensus was reached.}. In practice, different nodes may have different views of the blockchain and hence perceive proposed block differently. Yet, on closer inspection, this assumption can still be supported in the current framework: if a node is honest but has a (completely) different view of the canonical chain due to (say) poor connection to the network, then their votes do not contribute to extending the canonical chain and indeed can be regarded as incorrect. For instance, in Ethereum, which is the base case for this paper, \emph{valid--invalid} votes are well-defined and identified \cite{ryan2018eip,Da18}.

\subsubsection*{Updating scheme \& Computational Complexity}
Dealing with the issue of valid-invalid blocks becomes more relevant in the design of the updating scheme. Clearly, faults that can be identified as deliberate should be treated differently than accidental ones. For instance, a validator who has honestly participated in the protocol for a long period of time and drops offline for a short period of time, should be able to quickly recover their previous voting profile. This motivates updating profiles by two variables $s_{i,t}$ and $l_{i,t}$ representing \emph{short-term} and \emph{long-term} voting respectively. In general, the advantages of the here employed generalized MWU algorithm -- e.g., convergence rates and tight bounds on its overall performance \cite{Ar12} -- can be further exploited alongside stability properties of opinion dynamics in networks \cite{Ma17} to yield more robust results also in the present context. Finally, the computational complexity of initializing, storing and updating validators' profiles does not exceed the complexity of performing the same functions on validators' stakes and hence it is not expected to have a significant impact on the overhead of any consensus that keeps track of validators' stakes. This intuition is even more likely to materialize within the framework of most state of the art proposals, like EOS.IO, Casper and Ouroboros that promote consensus systems with a limited number delegates, staking pools or potential validators. 

\subsubsection*{Entry \& threshold voting profiles}
\label{sub:initiate}
The levels at which voting profiles are initialized and suspended are crucial, since they can incentivize or prevent \emph{Sybil attacks}. The exact parametrization can be case-depend, justified by simulations or incorporate prior information for each entering validator, e.g., reputable financial institution versus unknown private staker. The present choice to initialize voting profiles at $0.5$ and to require that they remain in $[0.5,1)$ for all $t\ge g>0$, where $g$ denotes a potential initial grace period, is supported by both intuitive and theoretical arguments. From a mathematical perspective, the initial choice of $0.5$ represents an uninformative prior on an entering's validators voting profile. Similarly, the reason for the upper bound is purely numerical, namely to avoid the instability in $\ln{\(p_{i,t}/\(1-p_{i,t}\)\)}$ if $p_{i,t}=1$. In contrast, the suspension of validators with voting profile -- or probability of making a correct decision -- lower than $0.5$ is more fundamental. \cite{Be97} provide a detailed probabilistic argument to explain that such voters are harmful to consensus and their votes should not be considered. Moreover, in the specific context of distributed networks, allowing nodes to participate with scores lower than their initial one triggers \emph{Sybil attacks}, since it motivates switching to new or creating multiple accounts. Finally, suspending validators in terms of their voting behavior relaxes the need for frequent economic penalties \cite{Bu1811}. This makes the PoS ecosystem more secure and appealing to investors who would otherwise be concerned of suffering losses due to accidental misbehavior, e.g., dropping off-line or being censored.

\subsubsection*{Future implementation}
Currently, the proposed scheme seems more attractive for systems with low numbers of staking nodes: these may be permissioned and delegated PoS blockchains or blockchains with fixed number of stake pools. More closely related to this spirit are the proposals of \cite{Bu1810,ouro17,An18,Li18}. In these settings, the computational complexity of implementing a profiling system is not an issue. However, this is also expected to remain true in the general case of permissionless blockchains, since extra data storage is limited to one additional value per validator and computations to update profiles are linear in the size of the committees. Light on this issues will be shed by future implementations on simulated blockchains and if successful, on properly designated testnets of these blockchains. In the core of these studies will be the understanding of issues related to network conditions (latency, scalability), computational complexity to store and retrieve profiles, indirect economic effects on cryptocurrency prices \cite{Kok20}, and the testing of different updating schemes in terms of their convergence rates and efficiency bounds.

%

\section{Summary \& Conclusions}\label{sec:conclusions}

Existing PoS protocols select staking nodes proportionally to their stake to form block-creating committees. Yet, they do not guarantee that selected committees will create blocks, since consensus may fail due to accidental or adversarial behavior. Thus, the perceived fairness in the distribution of rewards in proportion to the stake of participating nodes is actually violated. Motivated by this observation, we studied \emph{weighted voting} as a way to improve the consensus mechanism. We introduced validators' voting profiles -- that quantify the probability that a validator will cast a correct vote based on their so far contribution to the protocol -- and defined the proper mathematical framework to apply the results of \cite{Be97} on optimal decision rules in committee voting. Using the approach of \cite{Ar12}, we designed a multiplicative weights algorithm to update individual validator's profiles according to their voting behavior, the consensus outcome and collective blockchain welfare. The result is a two-layered scheme in which selection of nodes and allocation of rewards are performed by the underlying PoS mechanism whereas blocks are decided by a weighted majority voting rule. This scheme improves consensus \emph{within} selected committees by scaling votes according to validators' profiles without interfering with the PoS execution. Hence, it can be tested, implemented and reverted with minimal cost to existing users. On the negative side, the introduction of a profiling scheme in a distributed network raises new risks associated with manipulation of the relevant information. We discussed such risks along with actions that should be considered in the design of future PoS protocols.

\bibliographystyle{abbrv}
\bibliography{voting_bib}

\end{document}